\documentclass[letterpaper,11pt]{article}
\usepackage{authblk}
\usepackage[symbol]{footmisc}
\usepackage[bold]{complexity}
\usepackage{amsmath}
\usepackage{amssymb}
\usepackage{fullpage}
\usepackage{arydshln}
\usepackage{amsthm}
\usepackage{thmtools}
\usepackage{xcolor}
\usepackage{enumitem}
\newtheorem{lemma}{Lemma}

\newtheorem{theorem}{Theorem}

\theoremstyle{definition}
\newtheorem{definition}{Definition}
\newtheorem{example}{Example}
\newcommand{\sjarp}{\sharp}
\newcommand{\sharpCQA}[1]{\mathsf{\sjarp CERTAINTY}({#1})}
\newcommand{\homosymbol}[1]{{\mathsf{enc}}_{#1}} 
\newcommand{\newhomosymbol}[1]{{\mathsf{enc}}^{\ast}_{#1}} 
\newcommand{\homo}[2]{\homosymbol{#2}({#1})} 
\newcommand{\newhomo}[2]{\newhomosymbol{#2}({#1})} 
\newcommand{\bfS}{\mathbf{S}} 
\newcommand{\db}{{\mathbf{db}}}

\newcommand{\onwaar}{{\mathbf{false}}}
\newcommand{\card}[1]{|{#1}|}

\newcommand{\skBCQ}{\mbox{\sf skBCQ}}
\newcommand{\cxBCQ}{\mbox{\sf cxBCQ}}
\newcommand{\signature}[2]{[{#1},{#2}]}
\newcommand{\keyequal}{\sim}
\newcommand{\rep}{\mathbf{r}}

\newcommand{\repairs}[1]{\mathsf{rset}({#1})}
\newcommand{\defeq}{\mathrel{\mathop:}=}
\newcommand{\queryvars}[1]{\mathsf{vars}({#1})}
\newcommand{\sequencevars}[1]{\mathsf{vars}({#1})}
\newcommand{\inverse}[1]{{#1}^{-1}}
\newcommand{\komtna}{\circ}
\newcommand{\formula}[1]{\left({#1}\right)}
\newcommand{\couple}[2]{\langle{#1},{#2}\rangle}
\title{Corrigendum to \emph{``Counting Database Repairs that Satisfy Conjunctive Queries with Self-Joins''}}
\author{Jef Wijsen}
\affil{University of Mons, Belgium}
\date{}
\begin{document}
\maketitle

\begin{abstract}
The helping Lemma~7 in~[Maslowski and Wijsen, ICDT, 2014] is false.
The lemma is used in (and only in) the proof of Theorem~3 of that same paper.
In this corrigendum, we provide a new proof for the latter theorem.
\end{abstract}

\section{The Flaw}

The helping Lemma~7 in~\cite{DBLP:conf/icdt/MaslowskiW14} is false.
A counterexample is given next.

\begin{example}\label{ex:flaw}
For $\bfS=\{R,S\}$ and $q=\{R(\underline{x},y)$, $S(\underline{y})\}$,
we have
$\homo{q}{\bfS}=\{N(\underline{R,x},y)$, $N(\underline{S,y},0)\}$.
From~\cite[Lemma~8]{DBLP:conf/icdt/MaslowskiW14}, it follows that
$\sharpCQA{\homo{q}{\bfS}}$ is $\sharp\P$-hard.
From~\cite[Theorem~4]{DBLP:journals/jcss/MaslowskiW13},
it follows that $\sharpCQA{q}$ is in $\FP$.
Consequently, assuming $\sharp\P\neq\FP$,
there exists no polynomial-time many-one reduction from  $\sharpCQA{\homo{q}{\bfS}}$ to  $\sharpCQA{q}$.
Lemma~7 in~\cite{DBLP:conf/icdt/MaslowskiW14} is thus false.
\qed
\end{example}

The first part in the proof of Lemma~7 in~\cite{DBLP:conf/icdt/MaslowskiW14} is correct; it shows a polynomial-time many-one reduction from $\sharpCQA{q}$ to $\sharpCQA{\homo{q}{\bfS}}$.
However, the second part in that proof is flawed when it claims \emph{``We can compute in polynomial time the (unique) database $\db_{0}'$ with schema $\bfS$ such that $\homo{\db_{0}'}{\bfS}=\db_{0}$.''}
The flaw is that the database $\db_{0}'$ does not generally exist, as shown next.
Let $\bfS=\{R,S\}$ and $q=\{R(\underline{x},y)$, $S(\underline{y})\}$, as in Example~\ref{ex:flaw}.
Then, $\homo{q}{\bfS}=\{N(\underline{R,x},y)$, $N(\underline{S,y},0)\}$.
A legal input to $\sharpCQA{\homo{q}{\bfS}}$ is $\db_{0}=\{N(\underline{R,b},c)$, $N(\underline{S,c},0)$, $N(\underline{S,c},1)\}$.
However, there exists no database $\db_{0}'$ such that $\homo{\db_{0}'}{\bfS}=\db_{0}$.
Indeed, for every database $\db_{0}'$ with schema $\bfS$, if $N(\underline{S,c},s)\in\homo{\db_{0}'}{\bfS}$, then $s=0$. 

\section{The Solution}

The following treatment is relative to a database schema~$\bfS$.
Let $k,m$ be non-negative integers such that every relation name in $\bfS$ has at most $k$ primary-key positions, and at most $m$ non-primary-key positions.
We define a new function $\newhomo{q}{\bfS}$ which encodes Boolean conjunctive queries~$q$ into unirelational Boolean conjunctive queries. 
For $\newhomo{q}{\bfS}$,
we use a fresh relation name $N$ with $k+1$ primary-key positions,
and $m$ non-primary-key positions.
For every atom $R(\underline{\vec{x}},\vec{y})$ in $q$,
the query $\newhomo{q}{\bfS}$ will contain some atom 
$N(\underline{R,\vec{x},\vec{0}},\vec{y},\vec{z})$,
where $\vec{0}$ is a sequence of padding zeros,
and $\vec{z}$ is a sequence of padding fresh variables, all distinct and not occurring elsewhere.
This encoding is different from~\cite[Definition~3]{DBLP:conf/icdt/MaslowskiW14} where a sequence of padding zeros was used instead of~$\vec{z}$.

\begin{example}\label{ex:padding}
We illustrate the difference between the old encoding $\homo{\cdot}{\bfS}$ of~\cite[Definition~3]{DBLP:conf/icdt/MaslowskiW14} and the newly proposed encoding $\newhomo{\cdot}{\bfS}$.
For $q_{0}=\{R(\underline{x},y)$, $S(\underline{y})\}$,
we have
\begin{eqnarray*}
\homo{q_{0}}{\bfS} & = & \{N(\underline{R,x},y), N(\underline{S,y},0)\},\\
\newhomo{q_{0}}{\bfS} & = & \{N(\underline{R,x},y), N(\underline{S,y},z)\}.
\end{eqnarray*}
We recall from~\cite[p.~156]{DBLP:conf/icdt/MaslowskiW14} that the \emph{complex part} of a Boolean conjunctive query contains every atom~$F\in q$ such that some non-primary-key position in $F$  contains either a variable with two or more occurrences in $q$ or a constant. 
Note that $N(\underline{S,y},0)$ belongs to the complex part of $\homo{q_{0}}{\bfS}$, while $N(\underline{S,y},z)$ is not in the complex part of $\newhomo{q_{0}}{\bfS}$.
\qed
\end{example}

\begin{definition}
We define $\skBCQ$ as the class of Boolean conjunctive queries in which all relation names are simple-key.
We say that a query $q\in\skBCQ$ is \emph{minimal} if both
\begin{itemize}
\item
$q$ contains no two distinct atoms $R_1(\underline{x_{1}},\vec{y}_{1})$, $R_2(\underline{{x}_{2}},\vec{y}_{2})$ such that $R_{1}=R_{2}$ and $x_{1}=x_{2}$; and
\item
there exists no substitution $\theta$ over $\queryvars{q}$ such that $\theta(q)\subsetneq q$.
\end{itemize} 
We define $\cxBCQ$ as the class of \emph{unirelational} Boolean conjunctive queries~$q$ whose relation name has signature $\signature{n}{2}$ (for some $n\geq 2$) such that for every $F\in q$,
the first position of $F$ is a constant.
\end{definition}


\begin{definition}
The \emph{intersection graph} of a Boolean conjunctive query is an undirected graph whose vertices  are the atoms of~$q$.
There is an undirected edge between any two atoms that have a variable in common.
\end{definition}

\begin{lemma}\label{lem:reuse}
Assume $\sharp\P\neq\FP$.
For every minimal query $q$ in $\skBCQ$,
if $\sharpCQA{\newhomo{q}{\bfS}}$ is $\sharp\P$-hard, then so is $\sharpCQA{q}$.
\end{lemma} 
\begin{proof}
Let $q$ be a minimal query in $\skBCQ$ such that $\sharpCQA{\newhomo{q}{\bfS}}$ is $\sharp\P$-hard.
Note that $q$ does not need to be unirelational or self-join-free.
The query $\newhomo{q}{\bfS}$, which is unirelational, is a legal input to the function IsEasy of~\cite[p.~163]{DBLP:conf/icdt/MaslowskiW14}.\footnote[2]{For uniformity of notation, we will assume that the unirelational query uses relation name~$N$.}
Since $\sharpCQA{\newhomo{q}{\bfS}}$ is $\sharp\P$-hard, the function IsEasy will return $\onwaar$ on input $\newhomo{q}{\bfS}$. 
This function will repeat, as long as possible, the following step: pick some atom $N(\underline{R,c},\vec{y})$ and some variable $y\in\sequencevars{\vec{y}}$, with $R$ some relation name (treated as a constant) and $c$ some constant, and replace all occurrences of~$y$ with an arbitrary constant.
Let $\bar{q}$ be the query that results from these steps.
Clearly, for every atom $N(\underline{R,s},\vec{t})$ in~$\bar{q}$,
either $s$ is a constant or $\vec{t}$ is variable-free.
Since IsEasy returns $\onwaar$ on input $\bar{q}$, it follows that~$\bar{q}$ does not satisfy the premise of~\cite[Lemma~5]{DBLP:conf/icdt/MaslowskiW14}.
Therefore, it must be the case that $\bar{q}$ contains two distinct atoms $N(\underline{R,x},\vec{u})$ and
$N(\underline{S,y},\vec{w})$
that are connected in the intersection graph of $\bar{q}$ such that 
\begin{itemize}
\item
$R$ and $S$ are relation names (serving as constants), not necessarily distinct;
\item
$x$ and $y$ are distinct variables; and
\item
neither $\vec{u}$ nor $\vec{w}$ is exclusively composed of variables occurring only once in the query.
That is, $N(\underline{R,x},\vec{u})$ and
$N(\underline{S,y},\vec{w})$ belong to the complex part of~$\bar{q}$.
\end{itemize}
For every relation name $R$ that appears in $q$,
we assume fresh relation names $R_1, R_2, R_3,\ldots$ with the same signature  as $R$.
Using these relation names, we can construct a self-join-free Boolean conjunctive query $q'$ such that $\card{q'}=\card{q}$ and for every atom $R(\underline{x},\vec{y})$ in $q$,
the query $q$ contains some atom $R_{i}(\underline{x},\vec{y})$. 
For example, if $q=\{R(\underline{x},y)$, $R(\underline{y},z)$, $S(\underline{z},x)\}$,
then we can let 
$q'=\{R_1(\underline{x},y)$, $R_2(\underline{y},z)$, $S_1(\underline{z},x)\}$.
It can now be shown that the function IsSafe in~\cite[p.~158]{DBLP:conf/icdt/MaslowskiW14} will return $\onwaar$ on input $q'$, and thus $\sharpCQA{q'}$ is $\sharp\P$-hard.
Indeed, whenever IsEasy picked $N(\underline{R,c},\vec{y})$ and some variable $y\in\sequencevars{\vec{y}}\cap\queryvars{q}$,
the function IsSafe can execute SE3 on the corresponding $R_{i}$-atom of $q'$.
This eventually leads to a query whose complex part contains two atoms $R_{i}(\underline{x},\vec{u}')$ and  $S_{j}(\underline{y},\vec{w}')$, $x\neq y$, that are connected in the intersection graph,
at which point IsSafe will return $\onwaar$.
In this reasoning, one needs that non-primary-key positions are padded with fresh variables occurring only once, as can be seen from Example~\ref{ex:padding}.

In the remainder of this proof, we show the existence of a polynomial-time many-one reduction from $\sharpCQA{q'}$ to $\sharpCQA{q}$. 
We incidentally note that the remaining reasoning, which generalizes the proof of~\cite[Lemma~2]{DBLP:conf/icdt/MaslowskiW14}, does not require that relation names are simple-key.

Let $f$ be a mapping from facts to facts such that 
for every atom $R_i(x_1,\dots,x_n)\in q'$,
for every $R_i$-fact $A\defeq R_i(a_1,\dots,a_n)$,
$f(A)\defeq R(\couple{a_1}{x_1},\dots,\couple{a_n}{x_n})$.
Notice that $f$ maps $R_i$-facts to $R$-facts.
Here, every couple $\couple{a_i}{x_i}$ denotes a constant such that $\couple{a_i}{x_i}=\couple{a_j}{x_j}$ if and only if both $a_i=a_j$ and $x_i=x_j$.
Moreover, if $c$ is a constant, then $\couple{c}{c}\defeq c$.
Since no two distinct atoms of~$q$ agree on both their relation name and primary key,
it will be the case that for all facts $A$ and $B$,
$A\keyequal B$ if and only if $f(A)\keyequal f(B)$,
where $\keyequal$ denotes ``is key-equal-to.''

We extend the function $f$ in the natural way to databases $\db$ that use only relation names from~$q'$:
$f(\db)\defeq\{f(A)\mid A\in\db\}$.
Clearly, $f(\db)$ can be computed in polynomial time in the size of $\db$.
Let $\db$ be a set of facts with relation names in $q'$.
It can be easily seen that $\card{\repairs{\db}}=\card{\repairs{f(\db)}}$ and $\repairs{f(\db)}=\{f(\rep)\mid\rep\in \repairs{\db}\}$.
Let $\rep$ be an arbitrary repair of $\db$.
It suffices to show that $$\rep\models q'\iff f(\rep)\models q.$$

For the implication $\implies$, assume that $\rep\models q'$.
We can assume a valuation $\theta$ over $\queryvars{q'}$ such that $\theta(q')\subseteq\rep$.
Let $\mu$ be the valuation such that for every variable $x\in\queryvars{q'}$, $\mu(x)=\couple{\theta(x)}{x}$. 
By our construction of $q'$ and $f$, it will be the case that $\mu(q)\subseteq f(\rep)$, thus $f(\rep)\models q$.

For the implication $\impliedby$, assume that $f(\rep)\models q$.
We can assume a valuation~$\theta$ over $\queryvars{q}$ such that $\theta(q)\subseteq f(\rep)$.
Notice that if $c$ is a constant in $q$,
then it must be the case that $\theta(c)=\couple{c}{c}\defeq c$.
We define $\theta_L$ as the substitution that maps every variable $x$ in $\queryvars{q}$ to the first coordinate of $\theta(x)$;
and $\theta_R$ maps every $x$ to the second coordinate of $\theta(x)$.
It is convenient to think of $L$ and $R$ as references to the Left and the Right coordinates, respectively.
Thus, by definition, $\theta(x)=\couple{\theta_L(x)}{\theta_R(x)}$.

By inspecting the right-hand coordinates of couples $\couple{a_{i}}{x_{i}}$ in $f(\rep)$, it can be easily seen that $\theta(q)\subseteq f(\rep)$ implies $\theta_R(q)\subseteq q$.
Since the query $q$ is minimal,
it follows that $\theta_R(q)=q$, i.e., $\theta_R$ is an automorphism.
Since the inverse of an automorphism is an automorphism,
$\inverse{\theta_R}$ is an automorphism as well.
Note that $\theta_R$ will be the identity on constants that appear in $q$.
We now define $\mu\defeq\theta_L\komtna\inverse{\theta_R}$ (i.e.,
$\mu$ is the composed function $\theta_L$ after the inverse of $\theta_R$),
and show that $\mu(q')\subseteq\rep$, which implies the desired result that $\rep\models q'$.
To this extent, let $R_i(x_1,\dots,x_n)$ be an arbitrary atom of $q'$.
It suffices to show $R_i(\mu(x_1),\dots,\mu(x_n))\in\rep$,
which can be proved as follows.
From  $R_i(x_1,\dots,x_n)\in q'$,
it follows  $R(x_1,\dots,x_n)\in q$.
Thus, since $\inverse{\theta_R}$ is an automorphism,
$$R\formula{\enskip
\inverse{\theta_R}(x_1),\enskip
\dots,\enskip
\inverse{\theta_R}(x_n)
\enskip}\in q.$$
Since $\theta(q)\subseteq f(\rep)$,
$$R\formula{\enskip
\theta\formula{\inverse{\theta_R}(x_1)},\enskip
\dots,\enskip
\theta\formula{\inverse{\theta_R}(x_n)}
\enskip}\in f(\rep).$$
Since, for every symbol~$s$, $\theta(s)=\couple{\theta_L(s)}{\theta_R(s)}$ and $\theta_R\formula{\inverse{\theta_R}(s)}=s$,
we obtain
$$R\formula{\enskip
\couple{\theta_L(\inverse{\theta_R}(x_1))}{x_1},\enskip
\dots,\enskip
\couple{\theta_L(\inverse{\theta_R}(x_n))}{x_n}
\enskip}\in f(\rep).$$
That is, by our definition of $\mu$,
$$R\formula{\enskip
\couple{\mu(x_1)}{x_1},\enskip
\dots,\enskip
\couple{\mu(x_n)}{x_n}
\enskip}\in f(\rep).$$
From this, it is correct to conclude that
$R_i(\mu(x_1),\dots,\mu(x_n))\in\rep$.
This concludes the proof.
\end{proof}

\begin{lemma}\label{lem:forward}
For every Boolean conjunctive query~$q$,
there exists a polynomial-time many-one reduction from $\sharpCQA{q}$ to $\sharpCQA{\newhomo{q}{\bfS}}$.
\end{lemma}
\begin{proof}
Let $q$ be a Boolean conjunctive query.
Let $R$ be a relation name that occurs in~$q$.
Let $\{R(\underline{\vec{x}_{i}},\vec{y}_{i})\}_{i=1}^{m}$ be the set of $R$-atoms of~$q$. 
Then, $\newhomo{q}{\bfS}$ will contain, for every $i\in\{1,\dots,m\}$, some atom 
$N(\underline{R,\vec{x}_{i},\vec{0}},\vec{y}_{i},\vec{z}_{i})$,
where $\vec{z}_{i}$ is a (possibly empty) sequence of distinct fresh variables not occurring elsewhere.
For every $R$-fact $A\defeq R(\underline{\vec{a}},\vec{b})$, we define $f(A)\defeq N(\underline{R,\vec{a},\vec{0}},\vec{b},\vec{0})$.
Note here that $f(A)$ depends on the signatures of~$R$ and~$N$, but not on the $R$-atoms of $q$.
The mapping~$f$ is defined similarly for all relation names that appear in~$q$.
It can be easily seen that for all facts $A$ and $B$ whose relation names appear in~$q$,
$A\keyequal B$ if and only if $f(A)\keyequal f(B)$.

If $\db$ is an instance of $\sharpCQA{q}$,
we can assume without loss of generality that every relation name in~$\db$ also appears in~$q$.
We extend the function $f$ to such instances $\db$ of $\sharpCQA{q}$:
$f(\db)\defeq\{f(A)\mid A\in\db\}$.
Obviously, $f(\db)$ can be computed in polynomial time in the size of~$\db$.
It is also obvious that $\card{\repairs{\db}}=\card{\repairs{f(\db}}$ and $\repairs{f(\db)}=\{f(\rep)\mid\rep\in\repairs{\db}\}$.
It suffices to show that for every repair $\rep$ of $\db$,
$$\rep\models q\iff f(\rep)\models\newhomo{q}{\bfS}.$$

For the implication $\implies$, assume $\rep\models q$.
We can assume a valuation $\theta$ over $\queryvars{q}$ such that $\theta(q)\subseteq\rep$.
Let $\theta'$ be the valuation that extends $\theta$ from $\queryvars{q}$ to $\queryvars{\newhomo{q}{\bfS}}$ such that $\theta'(z)=0$ for every variable $z$ that appears in $\newhomo{q}{\bfS}$ but not in $q$.
By the construction of $f$, it will be the case that $\theta'(\newhomo{q}{\bfS})\subseteq f(\rep)$.
Indeed, if $\newhomo{q}{\bfS}$ contains $N(\underline{R,\vec{x}_{i},\vec{0}},\vec{y}_{i},\vec{z}_{i})$,
then $\rep$ will contain $R(\underline{\theta(\vec{x}_{i})},\theta(\vec{y}_{i}))$,
hence $f(\rep)$ will contain $N(\underline{R,\theta'(\vec{x}_{i}),\vec{0}},\theta'(\vec{y}_{i}),\vec{0})$ and $\theta'(\vec{z}_{i})=\vec{0}$.

For the implication $\impliedby$, assume $f(\rep)\models\newhomo{q}{\bfS}$.
We can assume a valuation $\theta$ over $\queryvars{\newhomo{q}{\bfS}}$ such that $\theta(\newhomo{q}{\bfS})\subseteq f(\rep)$.
It is straightforward to see that $\theta(q)\subseteq\rep$.
\end{proof}

We now give the new proof for Theorem~3 in~\cite{DBLP:conf/icdt/MaslowskiW14}.

\begin{theorem}[\mbox{\cite[Theorem~3]{DBLP:conf/icdt/MaslowskiW14}}]
The set $\{\sharpCQA{q}\mid q\in\skBCQ\}$ exhibits an effective $\FP$-$\sharp\P$-dichotomy.
\end{theorem}
\begin{proof}[New proof]
Let $q\in\skBCQ$.
It can be decided whether $q$ can be satisfied by a consistent database.
If~$q$ cannot be satisfied by a consistent database, then for every database $\db$, the number of repairs of $\db$ satisfying~$q$ is~$0$.
An example is $q=\{R(\underline{x},0)$, $R(\underline{x},1)\}$.
Assume next that $q$ can be satisfied by a consistent database.
Then, we can compute a minimal query~$q_{m}$ such that for every database, the number of repairs satisfying $q_{m}$ is equal to the number of repairs satisfying $q$.
That is, the problems $\sharpCQA{q_{m}}$ and $\sharpCQA{q}$ are identical.

Then, $\newhomo{q_{m}}{\bfS}$ belongs to $\cxBCQ$.
By~\cite[Lemma8]{DBLP:conf/icdt/MaslowskiW14}, the set 
$\{\sharpCQA{q}\mid q\in\cxBCQ\}$ exhibits an effective $\FP$-$\sharp\P$-hard dichotomy.
If the problem $\sharpCQA{\newhomo{q_{m}}{\bfS}}$ is in $\FP$,
then $\sharpCQA{q}$ is in $\FP$ by Lemma~\ref{lem:forward};
and if $\sharpCQA{\newhomo{q_{m}}{\bfS}}$ is $\sharp\P$-hard,
then $\sharpCQA{q}$ is $\sharp\P$-hard by Lemma~\ref{lem:reuse}.
Consequently, $\sharpCQA{q}$ is in $\FP$ or $\sharp\P$-hard,
and it is is decidable which of the two cases applies.
\end{proof}

\end{document}